\documentclass[11pt]{article}
\usepackage[paper=a4paper, margin=1in]{geometry}

\usepackage[T1]{fontenc}

\usepackage{graphicx}

\usepackage{amssymb}
\usepackage{amsmath}
\usepackage{amsthm}
\usepackage{appendix}
\usepackage{booktabs}
\usepackage{url}
\usepackage[linesnumbered,ruled,vlined]{algorithm2e}

\newtheorem{theorem}{Theorem}
\newtheorem{lemma}[theorem]{Lemma}  % 与 theorem 共用编号
\newtheorem{definition}{Definition}

\begin{document}

\title{
Prediction-Augmented Mechanism Design for Weighted Facility Location\thanks{An extended abstract of this paper is to appear in the 19th Annual Conference on Theory and Applications of Models of Computation (TAMC 2025).}
}

\author{
Yangguang Shi\thanks{School of Computer Science and Technology, Shandong University, Qingdao, China. \texttt{shiyangguang@sdu.edu.cn}. Department of Economics, University of Bath, Claverton Down, Bath, UK. ORCID: 0000-0003-2923-2748}
\and
Zhenyu Xue\thanks{School of Computer Science and Technology, Shandong University, Qingdao, China. \texttt{xuezhenyu@mail.sdu.edu.cn}. ORCID: 0009-0004-0868-2488}
}

\date{}

\maketitle 
\begin{abstract}
Facility location is fundamental in operations research, mechanism design, and algorithmic game theory, with applications ranging from urban infrastructure planning to distributed systems. Recent research in this area has focused on augmenting classic strategyproof mechanisms with predictions to  achieve an improved performance guarantee against the uncertainty under the strategic environment. Previous work has been devoted to address the trade-off obstacle of balancing the consistency (near-optimality under accurate predictions) and robustness (bounded inefficiency under poor predictions) primarily in the unweighted setting, assuming that all agents have the same importance. However, this assumption may not be true in some practical scenarios, leading to research of weighted facility location problems. 

The major contribution of the current work is to provide a prediction augmented algorithmic framework for balancing the consistency and robustness over strategic agents with non-uniform weights. In particular, through a reduction technique that identifies a subset of \emph{representative} instances and maps the other given locations to the representative ones, we prove that there exists a \emph{strategyproof} mechanism achieving a bounded consistency guarantee of $\frac{\sqrt{(1+c)^2W^2_{\min}+(1-c)^2W^2_{\max}}}{(1+c)W_{\min}}$ %sqrt((1+c)^2 * W_min^2 + (1-c)^2 * W_max^2) / ((1+c) * W_min)-consistency
 and a bounded robustness guarantee of $\frac{\sqrt{(1-c)^2W^2_{\min}+(1+c)^2W^2_{\max}}}{(1-c)W_{\min}}$ % sqrt((1-c)^2 * W_min^2 + (1+c)^2 * W_max^2) / ((1-c) * W_min)
in weighted settings, where $c$ can be viewed as a parameter to make a trade-off between the consistency and robustness and $W_{\min}$ and $W_{\max}$ denote the minimum and maximum agents' weight. We also prove that there is no strategyproof deterministic mechanism that reach $1$-consistency and  $O\left( n \cdot \frac{W_{\max}}{W_{\min}} \right)$-robustness in weighted FLP, even with fully predictions of all agents.

\vspace{1em}
\noindent\textbf{Keywords: }Weighted Facility Location, Prediction-Augmented Mechanism, Consistency and Robustness Trade-off, Upper Bound.
\end{abstract}

\newpage

\section{Introduction}
Facility Location Problem (FLP) is a classic problem in operations research, computer science, and economics \cite{chan2021mechanism}. 
Its goal is to choose one or more facility locations within a given space to optimize a global cost function between users (referred to as agents) and facilities. If only one facility needs to place, we say it single FLP\footnote{If not mentioned on purpose, in the following text, we say the \emph{FLP} as single FLP by default.}. In the problem, there are $n$ agents in a metric space $\mathcal{M}$, and each agent $i$ has her preferred location $p_{i}$ and must report a location $\hat{p}_i \in \mathcal{M}$ of the facility. We are required to collect this information to choose where to locate the facility $f \in \mathcal{M}$. Once we choose the facility location $f$, the individual cost of each agent $i$ is given by the distance $d(p_{i}, f)$ between $p_{i}$ and $f$, where $d(\cdot)$ is the distance function defined on $\mathcal{M}$. Our objective is to optimize a global function $g(P; f)$, where $P$ is the location profile of all the agents. Naturally, here comes to a question: \emph{What kind of global cost function we should consider?} Motivated by various applications in practical scenarios (see \cite{chan2021mechanism} for a more detailed overview), researchers in this area are often interested in the global objectives of (1) \emph{egalitarian social cost} and (2) \emph{utilitarian social cost}, where \emph{egalitarian social cost} is to the maximum cost over all the agents, and \emph{utilitarian social cost} is the average cost \cite{agrawal2022learning,good2020coordinate}. The study of the former is more difficult as applying predictions on any existing strategyproof mechanisms is hard, for the mechanisms focus on the bounding box of $P$ that prediction can hardly get involved. So the current work is devoted to the study of the latter. 

Under our setting, agents are selfish, so an agent $i$ may misreport her preferred location to minimize her individual cost, which means the reported location $\hat{p}_i$ is not necessarily equivalent to her real preferred location $p_i$. This leads to the \emph{strategic facility location problem}, where the focus shifts from purely optimizing a cost function to designing mechanisms that are strategyproof, i.e., mechanisms in which agents have no incentive to tell lie about their true preferences \cite{procaccia2013approximate,brandt2016handbook}. This problem is first studied by Procaccia and Tennenholtz \cite{procaccia2013approximate}, focusing on the egalitarian social cost, and Dekel \cite{dekel2010incentive} investigates how to optimize the utilitarian social cost. Procaccia and Tennenholtz \cite{procaccia2013approximate} prove that no deterministic and strategyproof mechanism can achieve better than a 2-approximation, when the goal is the egalitarian social cost, and Meir \cite{meir2019strategyproof} propose the coordinate median (CM) mechanism that reaches $\sqrt{n}$-approximation in $n$-dimensional case while optimizing the utilitarian social cost. These results establish the framework for subsequent research.

Research work on FLP has been applied in establishment of public services such as schools and hospitals \cite{jamshidi2009median,daskin1997network}. Existing work is often based on the unweighted assumption on the agents, which means that all the agents in the environment are equivalent in importance. Although this assumption reflects equal treatment in public services, this setting is uncommon in reality, where agents often differ in importance. Therefore, the unweighted assumption rarely holds in practice. For example, in the field of healthcare, existing work suggests that the older adults are often associated with a relatively higher weights \cite{ahmadi2017survey,micc2019healthcare}. Agents with distinct importance (weights) also exist in other applications including logistics planning \cite{boonmee2017facility}, disaster assistance \cite{carnero2025temporary}, etc. Obvliously, the unweighted setting of FLP cannot be adapted to these environments. That is why strategic FLP with non-uniform weights has attracted the attention of researchers \cite{zhang2014strategyproof,todo2011false}. Under the weighted setting, we assume that different weight to quantifies her importance to the decision maker: agents with higher weights take greater social cost and vice versa. 

Existing work on strategic FLP (unweighted or weighted) usually adopts the research paradigm of worst-case analysis. This paradigm focuses on extreme scenarios that rarely occur in practice, leading to algorithm design that is overly conservative for the average case. On the other hand, average-case analysis must deal with the uncertainty in real-world problems, such as the lack of knowledge about the true distribution of agent's preferred location in FLP. To address this issue, a novel research paradigm, \emph{prediction-augmented algorithm design} has recently emerged. It allows decision-makers to access external prediction sources to obtain forecasts about the uncertainty in the problem setting, and allows them to decide how much they trust the forecasts \cite{feldman2013strategyproof,escoffier2011strategy,meir2019strategyproof,zhang2014strategyproof}. When the prediction is accurate, we can trust it more and get a better resulting guarantee than using worst-case analysis, i.e., \emph{consistency} of the mechanism. On the other hand, when the prediction is arbitrarily wrong, we would like to ensure that the wrong prediction does not worsen the result too much and obtains a solid performance guarantee, i.e., \emph{robustness} of the mechanism \cite{lykouris2018competitive}. More precisely speaking, in the framework of prediction-augmented mechanism design of strategic FLP, given the agents' reported locations $P$, the designer is provided a predicted location $\hat{o}$, regarding the real optimal location $o(P)$. The designer can use the information to make a mechanism, whose result is $f(P, \hat{o})$. Given a social cost function $C$, the mechanism is $\alpha$-consistent when the prediction is accurate $(\hat{o} = o(P))$ if 
\begin{equation*}
    \max_P{\left\lbrace\frac{C(f(P, o(P)), P)}{C(o(P), P)}\right\rbrace} \leq \alpha
\end{equation*}
and $\beta$-robust even when the prediction is arbitrarily wrong if 
\begin{equation*}
    \max_{P, \hat{o}}{\left\lbrace\frac{C(f(P, \hat{o}), P)}{C(o(P), P)}\right\rbrace} \leq \beta.
\end{equation*}

\subsection{Our Contributions}

We investigate the prediction-augmented mechanism design of strategic FLP in a weighted setting, where the importance of the agents is non-uniform according to the weighted values. Our goal is to design deterministic, strategyproof mechanisms that incorporate a prediction of the optimal facility location, and achieve a tunable trade-off between consistency and robustness. 

Our contributions are as follows:

\begin{theorem}\label{thm1}
    There exists a deterministic and strategyproof mechanism in the weighted FLP with predictions that reaches $\frac{\sqrt{(1+c)^2W^2_{\min}+(1-c)^2W^2_{\max}}}{(1+c)W_{\min}}$-consistency and $\frac{\sqrt{(1-c)^2W^2_{\min}+(1+c)^2W^2_{\max}}}{(1-c)W_{\min}}$-robustness. 
\end{theorem}

$W_{\min}$ and $W_{\max}$ denote the minimum and maximum weights of agents, respectively. The mechanism uses a parameter \( c \in [0, 1) \) to control the degree of reliance on the prediction. When \( c \) is large, the mechanism places more trust in the prediction, achieving better consistency; when \( c \) is small, it becomes more conservative, improving the robustness. In particular, we emphasize that our performance guarantee matches the existing results of Zhang and Li \cite{zhang2014strategyproof} by setting $c=0$. Also, when $\frac{W_{\max}}{W_{\min}} = 1$, meaning that all agents' weights are equal, our results matches those from Agrawal \cite{agrawal2022learning} on the unweighted setting. 

\begin{theorem}\label{thm2}
    There is no deterministic and strategyproof mechanism in the weighted FLP with predictions that is $1$-consistent and  $O\left( n \cdot \frac{W_{\max}}{W_{\min}} \right)$-robust.
\end{theorem}

Theorem \ref{thm2} generalizes known impossibility results for the unweighted case and highlights how agents' weights fundamentally constrains achievable trade-offs. It indicates that the weighted case of facility location is more challenging than the unweighted one, as the performance guarantee of any mechanism designed for the weighted case can hardly get rid of the dependence on the input parameters $\frac{W_{\max}}{W_{\min}}$.

\subsection{Related Work}

\subsubsection{Strategic Facility Location.} Strategic FLP is a classic mechanism design problem that has been widely studied. For example, Procaccia and Tennenholtz \cite{procaccia2013approximate} research the approximate mechanism design without money and take this problem as an example. For single facility location in one dimension, there exists a deterministic mechanism that is 2-approximate in utilitarian social cost and optimal in maximum social cost by placing the facility on the medium of all agents. This mechanism also reaches the best approximate ratio among all deterministic and strategyproof mechanisms. Zhang and Li \cite{zhang2014strategyproof} expand this theory into weighted cases, and show that the coordinate medium mechanism achieves $\frac{W_{\max}}{W_{\min}}$-approximation in utilitarian social cost. Also, it is the best approximation ratio.

A lot of research work focus on this framework along with strategic FLP \cite{chen2024strategic,balkanski2024randomized}. Beyond the classic setting, extended researches include Obnoxious Facility Location (where agents seek maximum distance from facilities) \cite{lam2024proportional,drezner2020multiple} and Heterogeneous Facility Location (where facilities can be both desirable and undesirable, from the perspective of different agents). Results in this line of work allows agents to have different weights. However, the weight of each agent is constrained in a binary set $\lbrace -1, 1\rbrace$. In our work, the analysis is conducted precisely by removing this constraint.

\subsubsection{Prediction-Augmented Algorithms.}
Prediction-augmented algorithms (also known as the learning-augmented algorithms), is an algorithm design paradigm aiming to surpass the worst-case performance guarantee by allowing the algorithm to access predictions. Lykouris and Vassilvtiskii \cite{lykouris2021competitive} introduce consistency and robustness to analyze the performance of algorithms in this framework as the two primal metrics to measure a mechanism in this framework. After its release, many classic problems, like online paging problems \cite{lykouris2021competitive,jiang2022online}, scheduling problems \cite{purohit2018improving}, ski-rental problems \cite{wei2020optimal,wang2020online}, and secretary problems \cite{antoniadis2020secretary} have been researched in this framework. Facility location problems are also in the list. Fotakis and Gergatsouli \cite{fotakis2021learning} study online facility location problem with prediction, and Jiang and Liu \cite{jiang2021online} introduce dynamic prediction update to amend the prediction models. 

However, the problems mentioned above are not studied under a strategic setting. In strategic problems, the predictions focus on overcoming the information limitation of privately held information rather than future information \cite{agrawal2022learning}. Agrawal \cite{agrawal2022learning} first takes strategic FLP into research and designs Coordinate Medium (CM) mechanism with Prediction, also known as CMP mechanism. The CMP mechanism is simply adding $cn$ copies of $\hat{o}$ into $P$ to construct a new multiset $P'$ to run the CM mechanism, which outputs the median of $P'$ from each dimension separately. As is known that such median-point like mechanisms cannot be applied to general metric spaces \cite{schummer2002strategy}, Agrawal's work is conducted over the two-dimensional Euclidean space. Our work can be viewed as the extension of their work by adopting the same assumption on the metric space, while allowing agents to have non-uniform weights, which is more common in practice.

\section{Preliminaries}

\paragraph{Facility Location Problem (FLP).}

The classic FLP considers a set of agents located in a metric space (the current work focuses on the two-dimensional Euclidean space $\mathbb{R}^2$). Each agent $i \in [n]$ has a preferred location $p_i \in \mathbb{R}^2$. Given the preferred location set $P=\{p_1, \dots, p_n\}$, the goal of FLP is to choose a facility location $f \in \mathbb{R}^2$ that minimizes the average cost, measured as the average distance from all agents to the facility:
\[
C^u(f, P) = \frac{1}{n} \sum_{p_i \in P} d(f, p_i),
\]
where \(d(\cdot, \cdot)\) denotes the Euclidean distance. This is referred to as the \emph{utilitarian cost}. Given this social cost function, one can define and compute the optimal facility location, denoted as $o(P)$.

\paragraph{Strategic FLP.}

In strategic settings, the real preferred locations \(p_i\) are private information, and an agent may misreport it to reduce her individual cost $C(f, p_i) = d(f, p_i)$. A mechanism \(\mathcal{M}\) takes reported locations \((\hat{p}_1, \ldots, \hat{p}_n)\) and outputs a facility location \(f\). 
In order to prevent misreporting, we would like to ensure that the mechanism we design is \emph{strategyproof}. A mechanism is said to be strategyproof if the optimal decision of any agent is reporting her real information, regardless of others' reports. In our work, the reported information is about the real preferred locations, so the strategyproof constraints can be formulated as follows:
\[
C(f(P_{-i}, p_i), p_i) \leq C(f(P_{-i}, \hat{p}_i), p_i) \quad \forall P \in \mathbb{R}^{2n}, i \in [n],  \hat{p}_i \in \mathbb{R}^2.
\]

A well known mechanism designed for strategic FLP is the \emph{Coordinate Median (CM)} mechanism \cite{agrawal2022learning}, which selects the facility location by independently computing the median of all coordinates (if the number of agents $n$ is even, we assume that the smaller of two medians is returned for all coordinates). It proves that such a deterministic mechanism is strategyproof in two-dimensional space, and it guarantees a bounded approximation ratio of \(\sqrt{2}\) for the utilitarian objective \cite{agrawal2022learning}. The mechanism Generalized Coordinate Median (GCM) proposed by Agrawal \cite{agrawal2022learning} generalizes the classic version of CM by taking an independent and constant multiset $P'$ of locations (often called the \emph{phantom points}) as inputs parallel to the reported location set $P$. The location chosen by GCM for the facility is $CM(P \cup P')$. Similarly, the GCM mechanism is also proved to be deterministic and strategyproof \cite{agrawal2022learning}.

\paragraph{Strategic FLP With Weights.}

In many practical applications, agents differ in demand, importance, or service sensitivity. This difference can be modeled by assigning each agent a public positive weight \(w_i > 0\). In the weighted setting, the individual cost for agent \(i\) becomes:
\[
C(f, p_i) = w_i \cdot d(f, p_i),
\]
and the utilitarian cost becomes the average weighted distance:
\[
C^u(f, P) = \frac{1}{n} \sum_{p_i\in P} w_i \cdot d(f, p_i).
\]

The weighted case breaks many of the nice symmetry properties of the unweighted problem and presents additional difficulties for designing strategyproof mechanisms. For example, a single agent with an extremely high weight can significantly bias the optimal solution, or even dominate the outcome.

\paragraph{Prediction-augmented mechanism.}

For the weighted version of strategic FLP, we consider developing a prediction augmented mechanism that receives a prediction \(\hat{o} \in \mathbb{R}^2\) of the optimal location. This prediction may comes from some external machine learning models trained by the history data of agents. Our mechanism can be viewed as a functions $f(P, \hat{o})$ that maps the reported preferred locations and the prediction to a location of the facility. The goal of the current paper is to ensure that the proposed mechanism is strategyproof, while achieving a competitive ratio as close to $1$ as possible when the obtained predictions are highly accurate (i.e., providing a consistency guarantee), and maintaining an acceptable competitive ratio even when the prediction error rate is high (i.e., providing a robustness guarantee).

\section{Prediction Augmented Mechanism and Upper Bounds Analysis}

To prove Theorem \ref{thm1}, we first consider the CMP mechanism that is proved to be $\frac{\sqrt{2+2c^2}}{1+c}$-consistent and $\frac{\sqrt{2+2c^2}}{1-c}$-robust in the unweighted strategic FLP, and then try to extend it to the weighted case. The major contribution here is to prove that our extended version of CMP for the weighted setting has the same upper bound on the approximation ratio with the original version when choosing a proper value for $c$, a parameter for making trade-off between the consistency and robustness.

\subsection{Extending CMP to Weighted Case}
\paragraph{Vanilla CMP Mechanism.}

The vanilla CMP mechanism (i.e., the original version of CMP) is a simple way to combine agent reports with a predicted optimal location. It does this by taking $P'$ as the multiset of some copies of predicted location and applying it to GCM mechanism. The number of copies is controlled by a confidence parameter \(c \in [0, 1)\), more specifically, it adds $cn$ copies of $\hat{o}$.

If \(c\) is small, only a few copies of prediction are added, so the mechanism relies more on the reported information and is more robust to prediction errors. If \(c\) is large, more copies of prediction are added, and the mechanism more trusts the prediction.

Notably, when \(c = 0\), no copies of prediction are added, and CMP becomes equivalent to the classic CM mechanism.

This design allows CMP to smoothly trade off between being robust (when the prediction may be wrong) and following the prediction (when it is accurate), while remaining deterministic and strategyproof.

In this work, we extend the Coordinate Median with Prediction (CMP) mechanism to the weighted strategic setting, and analyze its consistency and robustness.

\begin{algorithm}[ht]
    \caption{Coordinate Median with Prediction (CMP)}
    \label{alg:cmp-unweighted}
    \KwIn{Reported agent locations \(P = \{p_1, \ldots, p_n\} = \{(x_1, y_1), \ldots, (x_n, y_n)\} \in \mathbb{R}^{2n}\); \\
    Prediction \(\hat{o} = (x_{\hat{o}}, y_{\hat{o}}) \in \mathbb{R}^2\); \\
    Confidence parameter \(c \in [0,1)\)}
    \KwOut{Facility location \(f \in \mathbb{R}^2\)}
    
    Let \(m \gets \lfloor cn \rfloor\)\;
    
    Let \(P' \gets P \cup m * \{\hat{o}\}\) (add \(m\) copies of \(\hat{o}\))\;
    
    Let \(x_{\text{med}} \gets\) \(med(x_1, \dots, x_n, x_{\hat{o}}, \dots, x_{\hat{o}})\) (median of the \(x\)-coordinate in \(P'\)) \;
    
    Let \(y_{\text{med}} \gets\) \(med(y_1, \dots, y_n, y_{\hat{o}}, \dots, y_{\hat{o}})\) (median of the \(y\)-coordinate in \(P'\)) \;
    
    \Return{\(f \gets (x_{\text{med}}, y_{\text{med}})\)}
    \end{algorithm}

\paragraph{Extension of CMP to Weighted Setting.}

We now analyze the Coordinate Median with Prediction (CMP) mechanism in the weighted setting introduced in our contribution section. Similar with the unweighted setting, heAs before, CMP still takes as input the reported agent locations and a predicted optimal location, adds phantom points at the prediction, and returns the coordinate-wise median over the augmented set. TThe procedure of our mechanism does not rely on the weights,consider weights when computing the facility location, and thus remains deterministic and strategyproofness by construction.

Although sharing a similar description of procedure with the vanilla CMP, the extended CMP asks for novel analysis techniques to incorporate the non-uniform weightsIn the weighted setting, the key difference lies in how performance is evaluated. In particular, eEach agent \(i\) is assigned a public weight \(w_i > 0\), and the cost function becomes:
\[
C^u(f, P) = \frac{1}{n} \sum_{p_i \in P} w_i \cdot d(f, p_i),
\]
reflecting the average weighted distance across agents.

Our main result, as stated in Theorem \ref{thm1}, shows that CMP achieves a consistency-robustness trade-off parameterized by \(c \in [0,1)\), with explicit performance bounds depending on the weight ratio \(W_{\max}/W_{\min}\). When theall weights are identicequal, these bounds reduce to the known results for the unweighted case~\cite{agrawal2022learning}. Particularly, when $c$ is set to $0$When \(c = 0\), the extended CMP reduces to the vanilla versionoriginal coordinate median mechanism and recovers the robustness bounds in~\cite{zhang2014strategyproof}.

We next analyze CMP's behavior on a structured input family, known as the Clusters-on-Axes (COA) construction proposed by Agrawal et al. \cite{agrawal2022learning}, to derive its consistency and robustness bounds and demonstrateunderstand why these bounds are tight.

\subsection{Consistency Analysis}

We begin our performance analysis of CMP in the weighted setting by focusing on its consistency: how well the mechanism performs when the prediction is accurate. Recall that consistency is measured by the worst-case ratio between the mechanism's output cost and the optimal cost, assuming that the prediction is exactly the optimal location.

To analyze this, we consider a structured class of worst-case instances, called \emph{Clusters-and-OPT-on-Axes} (COA), where some agents are at the top of the $y$-axis, and others split symmetrically on the left and right sides of the $x$-axis. These instances have been shown to yield tight bounds in the unweighted setting \cite{agrawal2022learning}, and we extend them here to match the weighted setting.

\begin{definition}[Weighted COA Instances]\label{def:COAW}
Given confidence parameter \( c \in [0, 1) \), a weighted COA instance consists of a reported loaction set $P$ and a prediction $\hat{o}$ that outputs \(f(P, \hat{o}, c) = (0,0)\) and optimal location \(o(P) = (0,1)\). Its agents positioned at \((0,1)\), \((x,0)\), or \((-x,0)\) for some \(x \geq 0\). Agents located at \((0,1)\) are assigned weight \(W_{\max}\), and agents at \((\pm x,0)\) are assigned weight \(W_{\min}\).
\end{definition}

This setup is designed to achieve the largest possiblmaximize the approximation ratio by concentrating high-weight agents near the optimal point and low-weight agents near the prediction. The following result formalizes that this setweighting indeed leads to the worst-case consistency.

\begin{theorem}\label{wcoa}
Among all weight distributions in COA instances, assigning weight \(W_{\max}\) to agents at \((0,1)\) and \(W_{\min}\) to agents at \((\pm x,0)\) maximizes the approximation ratio of CMP.
\end{theorem}

\begin{proof}
    The approximation ratio of such COA is
    \begin{equation*}
        A = \frac{mW_{\max} + (n-m)xW_{\min}}{(n-m)\sqrt{1+x^2}W_{\min}},
    \end{equation*}
    where $m$ is the amount of agents whose preferred location is $(0, 1)$.

    Now we let one agent $p$'s weight be $w \in (W_{\min}, W_{\max})$, the new approximation ratio become 
    \begin{equation*}
        B = \frac{w + (m-1)W_{\max} + (n-m)xW_{\min}}{(n-m)\sqrt{1+x^2}W_{\min}},
    \end{equation*}
    when $p \in \{(0, 1)\}$, or
    \begin{equation*}
        C = \frac{xw + mW_{\max} + (n-m-1)xW_{\min}}{\sqrt{1+x^2}w + (n-m-1)\sqrt{1+x^2}W_{\min}},
    \end{equation*}
    when $p \in \{(-x, 0), (x, 0)\}$.

    Since $w + (m-1)W_{\max} < mW_{\max}$, we have $B < A$, $C = \frac{x}{\sqrt{1+x^2}} + \frac{mW_{\max}}{\sqrt{1+x^2}w + (n-m-1)\sqrt{1+x^2}W_{\min}} < A$. So maximizing weights of agents in $(0, 1)$ and minimizing weights of agents in $(-x, 0)$ and $(x, 0)$ is the maximized weight distribution of COA.
\end{proof}

We now compute the consistency bound achieved by CMP under this worst-case weighted COA structure.

\begin{theorem}
For CMP with confidence \(c \in [0, 1)\), the worst-case consistency in weighted COA instances is
\[
\frac{\sqrt{(1+c)^2 W^2_{\min} + (1-c)^2 W^2_{\max}}}{(1+c) W_{\min}}.
\]
\end{theorem}

\begin{proof}
    To make approximation ratio as big as possible, we should maximize the amount of agents in $(0, 1)$. To guarantee the CMP output $(0, 0)$, there should be $\frac{(1-c)n}{2}$ agents report $(0, 1)$ at most to ensure the output of $y$-coordinate is $0$. In this instance, the approximation ratio is 
    \begin{equation*}
        \frac{\frac{1+c}{2}nxW_{\min}+\frac{1-c}{2}nW_{\max}}{\frac{1+c}{2}nW_{\min}\sqrt{1+x^2}} = \frac{\left(1+c\right)xW_{\min}+\left(1-c\right)W_{\max}}{\left(1+c\right)W_{\min}\sqrt{1+x^2}}
    \end{equation*}

    Take the first derivative with respect to $x$, we have
    \begin{equation*}
        \frac{\left(1+c\right)W_{\min}-\left(1-c\right)xW_{\max}}{\left(1+c\right)\left(1+x^2\right)^{3/2}W_{\min}}
    \end{equation*}

    Let numerator be $0$, we get $x = \frac{\left(1+c\right)W_{\min}}{\left(1-c\right)W_{\max}}$. For any $x < \frac{\left(1+c\right)W_{\min}}{\left(1-c\right)W_{\max}}$, the numerator is positive, and for any $x > \frac{\left(1+c\right)W_{\min}}{\left(1-c\right)W_{\max}}$, the numerator is negative, therefore when $x = \frac{\left(1+c\right)W_{\min}}{\left(1-c\right)W_{\max}}$, the approximation ratio reaches its maximum value.

    The maximum consistency is $\frac{\sqrt{(1+c)^2W^2_{\min}+(1-c)^2W^2_{\max}}}{(1+c)W_{\min}}$.
\end{proof}

This result demonstrates how the consistency of CMP degrades as the weight imbalance \(\frac{W_{\max}}{W_{\min}}\) increases. In the special case where all agents have equal weights, the bound reducsimplifies to the known result in the unweighted case. Now we begin to proveof why COA can achieve this bound.

In order for CMP to output the facility at \((0,0)\), there must be a sufficient number of agents near the origin to counterbalance the influence of phantom points at the prediction. The weighted COA structure achieves this by placing low-weight agents symmetrically on the horizontal axis, which allows for maximum geometric leverage while contributing minimally to the cost.

Meanwhile, high-weight agents are concentrated at the optimal point \((0,1)\), ensuring that the optimal facility location has minimal cost. This separation between the mechanism output and the optimal location—combined with the weight imbalance—leads to a maximal approximation ratio.

\begin{theorem}
Among all instances where CMP outputs the facility at \((0,0)\) and the prediction is accurate, the weighted COA construction achieves the highest possible consistency bound.
\end{theorem}

\begin{proof}
    Many lemmas used in the proof are the same as \cite{agrawal2022learning}. For those lemmas, we just provide them and do not give the detailed proofs.

    Before we begin the proof, we firstly present \emph{CA} and \emph{OA} proposed by Agrawal et al. \cite{agrawal2022learning}.
\end{proof}

\begin{definition}[Clusters-on-Axes (CA) (\cite{agrawal2022learning})]
Given a confidence value \(c\), and prediction \(\hat{o}\), define the family of multisets \(P\) such that:

\begin{enumerate}
    \item Output at origin: \(f(P, \hat{o}, c) = (0, 0)\);
    \item Opt in top-right quadrant: \(y_{o(P)} \geq x_{o(P)} > 0\);
    \item No move towards opt: for all \(p_i \in P\) and \(\varepsilon \in (0, 1]\),
    \[
    f\left((P \setminus \{p_i\}) \cup \{p_i + \varepsilon(o(P) - p_i)\}, \hat{o}, c\right) \neq f(P, \hat{o}, c);
    \]
    \item Cluster structure: there exist \(x_1, x_2, y_1, y_2 \geq 0\) such that:
    \begin{enumerate}
        \item Clusters on axes: For all \(p \in P\),
        \[
        p \in \{(-x_1, 0), (x_2, 0), (0, -y_1), (0, y_2), o(P)\};
        \]
        \item Less points in left: 
        \[
        |\{p \in P : p \in A^-_{x_{<}}\}| < |\{p \in P : p \in A^+_{x_{>}} \cup \{o(P)\}\}|;
        \]
        \item Less points in bottom: 
        \[
        |\{p \in P : p \in A^-_{y_{<}}\}| < |\{p \in P : p \in A^+_{y_{>}} \cup \{o(P)\}\}|;
        \]
        \item Symmetric x-clusters: If \((-x_1, 0), (x_2, 0) \in P\), then \(x_o + x_1 = x_2 - x_o\);
        \item Symmetric y-clusters: If \((0, -y_1), (0, y_2) \in P\), then \(y_o + y_1 = y_2 - y_o\).
    \end{enumerate}
\end{enumerate}

\end{definition}

\begin{definition}[Optimal-on-Axis (OA) (\cite{agrawal2022learning})]
Given a confidence value \(c\), and prediction \(\hat{o}\), define the family of multisets \(P\) such that:
\begin{enumerate}
    \item Output at origin: \(f(P, \hat{o}, c) = (0, 0)\);
    \item Opt on +\(y\)-axis: \(x_{o(P)} = 0, \ y_{o(P)} > 0\);
    \item Points on axes: \(\forall p \in P, \ p \in A_x \cup A_y\).
\end{enumerate}

\end{definition}

For the convenience of following proof, we make the following definitions like \cite{agrawal2022learning}:

We use \( P^C_{ca}(c) \), \( P^C_{oa}(c) \), and \( P^C_{coa}(c) \) to denote the families of CA, OA, and COA instances respectively, under confidence parameter \( c \in [0, 1) \), for the consistency analysis setting with accurate prediction \( \hat{o} = o(P) \).

Similarly, \( P^R_{ca}(c) \), \( P^R_{oa}(c) \), and \( P^R_{coa}(c) \) are used to the robustness analysis, setting with wrong prediction \( \hat{o} = (0,0) \).

\begin{lemma}
    For any point set \( P \) and confidence \( c \in [0,1) \) such that \\ 
\( f(P, \hat{o}(P), c) = (0,0) \),  
and \( y_o(P) \ge x_o(P) > 0 \),  
and for every \( p \in P \), we have \( p \in A_x \cup A_y \cup \{o(P)\} \),  
if at least one of the following two conditions holds:  
- \( |\{p \in P : p \in A^<_{-x}\}| \ge |\{p \in P : p \in A^>_{+x}\} \cup \{o(P)\}| \),  
- \( |\{p \in P : p \in A^<_{-y}\}| \ge |\{p \in P : p \in A^>_{+y}\} \cup \{o(P)\}| \),  
then there exists a point set \( Q \) such that \( r(Q, \hat{o}(Q), c) > r(P, \hat{o}(P), c) \),  
regardless of whether \( \hat{o}(P) = o(P), \hat{o}(Q) = o(Q) \) or \( \hat{o}(P) = \hat{o}(Q) = (0,0) \).

\end{lemma}

\begin{proof}
    Without loss of generality, assume that \( P \) satisfies  
\[
|\{p \in P : p \in A^<_{-x}\}| \ge |\{p \in P : p \in A^>_{+x}\} \cup \{o(P)\}|.
\]  
We construct a new point set \( Q \) by shifting all points not on the \( y \)-axis to the left by a small amount \( \epsilon \), while keeping the positions of the points on the \( y \)-axis unchanged. The value of \( \epsilon \) is small enough to ensure that \( o(Q) \) remains in the first quadrant.

It is clear that \( f(P,\hat{o}(P), c)=f(Q,\hat{o}(Q),c)=(0,0) \).  
For the points on the \( x \)-axis, we have  
\[
\sum_{p_i \notin A_y,\, p_i \in P}d(p_i, o(P))w_i = \sum_{q_i \notin A_y,\, q_i \in Q}d(q_i, o(Q))w_i.
\]  
For the points on the \( y \)-axis, we have  
\[
\sum_{p_i \in A_y,\, p_i \in P}d(p_i, o(P))w_i > \sum_{q_i \in A_y,\, q_i \in Q}d(q_i, o(Q))w_i,
\]  
i.e.,  
\[
C^u(o(P),P) > C^u(o(Q),Q).
\]

Since  
\[
|\{p \in P : p \in A^<_{-x}\}| \ge |\{p \in P : p \in A^>_{+x}\} \cup \{o(P)\}|,
\]  
and  
\[
C^u(f(P, \hat{o}(P), c)) - C^u(f(Q, \hat{o}(Q), c)) = \epsilon\left(\sum_{p_i \in A^<_{-x}}w_i - \sum_{p_i \in A^>_{+x}}w_i\right) \le 0,
\]  
it follows that  
\[
r(Q, \hat{o}(Q), c) > r(P, \hat{o}(P), c).
\]

\end{proof}

\begin{lemma}
    For any confidence \( c \in [0, 1) \), let \( \alpha = \max_{P \in P_{oa}^C(c) \cup P_{ca}^C(c)}r(P, o(P), c) \),\\ \( \beta = \max_{P \in P_{oa}^C(c) \cup P_{ca}^C(c)}r(P, (0, 0), c) \).  
    Then the CMP algorithm is \( \alpha \)-consistent and \( \beta \)-robust.
\end{lemma}

\begin{proof}
    Without loss of generality, let \( f(P, \hat{o}(P), c) = (0, 0) \) and \( y_o(P) \ge x_o(P) \ge 0 \).  
    Since \( f(P, \hat{o}(P), c) = (0, 0) \), it follows that \( f(P, (0, 0), c) = (0, 0) \),  
    hence \( r(P, \hat{o}(P), c) = r(P, (0, 0), c) \).

    If there exists a point \( p \in P \) such that \( p \notin A_x \cup A_y \cup \{o(P)\} \),  
    then it can be moved in the direction of \( o(P) \).  
    According to Lemma C1 in \cite{agrawal2022learning}, the geometric median remains unchanged during this movement.  
    Furthermore, by Lemma C2 in \cite{agrawal2022learning}, \( r(P, \hat{o}(P), c) \) increases during the movement.  
    Using this transformation, all points can be moved onto \( A_x \cup A_y \cup \{o(P)\} \);  
    denote the resulting point set as \( Q \), where \( o(Q) = o(P) \).

    If \( x_o(P) = 0 \), then \( Q \) is an OA instance.  
    Otherwise, by Lemma 4.10 in \cite{agrawal2022learning}, the points on the coordinate axes can be rearranged to  \\
    \( (-x_1, 0), (x_2, 0), (-y_1, 0), (y_2, 0) \), with \( x_1, x_2, y_1, y_2 > 0 \).  
    By the inequality of arithmetic and geometric means, assign the points on the \( x \)-axis by descending order of weights:  
    place the heavier half at \( (-x_1, 0) \), and the thinner half at \( (x_2, 0) \),  
    ensuring \( x_o(Q) + x_1 = x_2 - x_o(Q) \) and that \( \sum_{q_i \in A_x}|x_i|w_i \) remains unchanged.  
    (The same applies to the \( y \)-axis.)  
    Then \( C^u(Q, (0, 0)) \) remains unchanged, but \( C^u(Q, o(Q)) \) decreases, thus forming a CA instance.
    
    Using Lemma 1 and Lemma 2, we can transform any instances into either CA or OA without improving the approximation ratio. Then, we proof that both of them can be transformed into COA instances without improving the approximation ratio.

\end{proof}

\begin{lemma}
    For any confidence \( c \in [0,1) \) and any point set \( P \in P^C_{oa}(c) \),  
    there exists a point set \( Q \) such that \( r(Q, o(Q), c) > r(P, o(P), c) \)  
    or \( Q \in P^C_{coa}(c) \) with \( r(Q, o(Q), c) \ge r(P, o(P), c) \).  
    The conclusion also holds when the prediction is \( (0,0) \).
\end{lemma}

\begin{proof}
    Suppose \( P \in P^C_{oa}(c) \), then \( y_o(P) > 0 \),  
    and for every \( p \in P \), we have \( p \in A_x \cup A_y \).

    Define  
    \[
    d_x = \frac{1}{|\{p \mid p \in A_x\}|}\sum_{p \in A_x}|x_i|w_i,
    \]  
    and assign each of these points a weight equal to  
    \[
    \frac{1}{|\{p \mid p \in A_x\}|}\sum_{p \in A_x}w_i.
    \]  
    Place half of them at \( (-d_x, 0) \) and the other half at \( (d_x, 0) \).  
    By the Corollary 3 in \cite{good2020coordinate},  
    we conclude that \( C^U(P, o(P)) \) decreases while \( C^U(P,(0,0)) \) remains unchanged.

    For the points on the \( y \)-axis, consider an individual point \( p = (0, y) \),  
    and distinguish three cases:

    Case 1: \( y > y_o(P) \)  
    Move the point to \( o(P) \); during the process, both  
    \( C^U(P, o(P)) \) and \( C^U(P,(0,0)) \) decrease by the same amount,  
    and since \( C^U(P, o(P)) < C^U(P,(0,0)) \), we have that \( r(P, o(P), c) \) increases.

    Case 2: \( 0 \leq y < y_o(P) \)  
    Move the point to \( o(P) \); in this case,  
    \( C^U(P, o(P)) \) decreases and \( C^U(P,(0,0)) \) increases,  
    so \( r(P, o(P), c) \) increases.

    Case 3: \( y < 0 \)  
    First move the point to \( (0, 0) \) — this reduces to Case 2 —  
    then move it to \( o(P) \).

    In conclusion, all points on the \( y \)-axis can be moved to \( o(P) \)  
    to increase \( r(P, o(P), c) \). The resulting point set  
    \( Q \in P^C_{coa}(c) \), and the above conclusion similarly holds  
    when the prediction point is \( (0,0) \).
    
    Using Lemma 3, we can transform any OA instance into a corresponding COA instance without decreasing approximation ratio.
\end{proof}

\begin{lemma}
For any confidence \( c \in [0,1) \) and point set \( P \in P^C_{ca}(c) \), there exists a point set \( Q \) such that \( r(Q, o(Q), c) > r(P, o(P), c) \), or \( Q \in P^C_{ca}(c) \) and \( r(Q, o(Q), c) \ge r(P, o(P), c) \), and for any \( q \in Q \), \( q \in A_x \cup A_{+y} \cup \{o(Q)\} \). This also holds when the prediction result is \( (0,0) \).
\end{lemma}

\begin{proof}
Assume that \( P \in P^C_{ca}(c) \) is a multiset containing \( n \) points, and there exists \( p \in P \) such that \( p \in A^<_{-y} \).

According to the property of CA, for any \( p \in P \), we have:
\begin{itemize}
    \item \( p \in \{(-x_1, 0), (x_2, 0), (-y_1, 0), (y_2, 0), o(P)\} \),
    \item \( |\{p \mid p \in A_{-x}\}| < |\{p \mid p \in A_{+x} \cup \{o(P)\}\}| \),
    \item \( |\{p \mid p \in A_{-y}\}| < |\{p \mid p \in A_{+y} \cup \{o(P)\}\}| \).
\end{itemize}

Move the point \( p = (0, -y_1) \) to \( (-y_1, 0) \) to obtain a new point set \( P_1 \). Although \( o(P_1) \) may not satisfy \( y_{o(P_1)} \ge x_{o(P_1)} \ge 0 \), we still have \( f(P_1, o(P_1), c) = (0, 0) \), as shown below:

Since \( f(P, o(P), c) = (0, 0) \), in order for the CMP algorithm to produce this result, we require:
\[
|\{p \mid p \in A_{+x} \cup \{o(P)\}\}| + cn \le \frac{(1+c)n}{2}.
\]
Even if the \( cn \) optimal predicted points move to the left of the \( y \)-axis with \( o(P_1) \), it still holds that
\[
|\{p \mid p \in A_{-x}\}| + cn < \frac{(1+c)n}{2},
\]
so the CMP result in the \( x \)-axis is still 0; the same holds for the \( y \)-axis.

Thus,
\[
C^U(f(P, o(P), c), P) = C^U(f(P_1, o(P_1), c), P_1).
\]

Since \( y_{o(P)} \ge x_{o(P)} \ge 0 \), by the inequality of means:
\[
d((0, -y_1), o(P)) \ge d((-y_1, 0), o(P)),
\]
we have
\[
C^U(o(P_1), P_1) \le C^U(o(P), P_1) \le C^U(o(P), P).
\]

\begin{itemize}
    \item If \( o(P_1) \ne o(P) \), let \( Q = P_1 \); clearly, \( r(Q, o(Q), c) > r(P, o(P), c) \).
    \item If \( o(P_1) = o(P), x_1 \ne y_1 \), then by Lemma 4.10 in \cite{agrawal2022learning}, there exists \( Q \) such that\\ \( r(Q, o(Q), c) > r(P, o(P), c) \).
    \item If \( o(P_1) = o(P), x_1 = y_1 \), then
    \[
    d((0, -y_1), o(P)) = d((-y_1, 0), o(P)),
    \]
    hence
    \[
    C^U(o(P_1), P_1) = C^U(o(P), P_1) = C^U(o(P), P).
    \]
    We can continue to move each \( p \in A_{-y} \) in the same way to obtain \( Q \), satisfying
    \[
    Q \in P^C_{ca}(c),\quad r(Q, o(Q), c) \ge r(P, o(P), c),
    \]
    and for all \( q \in Q \),
    \[
    q \in A_x \cup A_y \cup \{o(Q)\}.
    \]
\end{itemize}
\end{proof}

\begin{lemma}[\cite{agrawal2022learning}]
For any confidence parameter $c \in [0,1)$ and any point set $P \in P^C_{ca}(c)$ with $p \in A_x \cup A_{+y} \cup \{o(P)\}$ for all $p \in P$, it holds that:
\[
|\{p \in P \mid p \in A_{-x}\}| \ge |\{p \in P \mid p = o(P)\}|.
\]
\end{lemma}

\begin{lemma}
For any confidence $c \in [0,1)$ and any point set $P \in P^C_{ca}(c)$ with all $p \in A_x \cup A_{+y} \cup \{o(P)\}$, there exists a point set $Q$ such that $r(Q, o(Q), c) > r(P, o(P), c)$ or $Q \in P^C_{ca}(c)$ and $r(Q, o(Q), c) \ge r(P, o(P), c)$.
\end{lemma}

\begin{proof}
From Lemma 5, we know
$|\{p \in P \mid p \in A_{-x}\}| \ge |\{p \in P \mid p = o(P)\}|$,
so we can iteratively move a point from $o(P)$ to the $y$-axis and a point from $A_{-x}$ to $A_{+x}$, until no point remains on $o(P)$.

We construct a new point set $Q = \{q_1, q_2, ..., q_n\}$ where:
- For some $x \ge 0$, $p_{i_1}=(-x, 0)$ is moved to $q_{i_1}=(x + 2x_o(P), 0)$;
- $p_{i_2} = (x_o(P), y_o(P))$ is moved to $q_{i_2} = (0, \sqrt{x_o(P)^2 + y_o(P)^2})$;
- For all $i \neq i_1, i_2$, set $q_i = p_i$.

Since $f(P, o(P), c) = (0, 0)$ and $P \in P^C_{ca}(c)$, we know
$|\{p \in A_{-x}\}| < |\{p \in A_{+x} \cup \{o(P)\}\}|$ and similarly in the $y$-axis, so
$|\{p \mid x_p < 0\}| < \frac{(1-c)n}{2}$, otherwise the CMP output would shift right.

Even if $o(Q)$ lies left of the $y$-axis, this inequality ensures $f(Q, o(Q), c) = (0, 0)$.

Let $A = \sum_{i \neq i_1} d(p_i, f(P, o(P), c)) w_i$ and $B = \sum_{i \neq i_2} d(p_i, o(P)) w_i$.

Assume $A + d(p_{i_1}, f(P, o(P), c)) w_{i_1} \le \sqrt{2} B$, and that
$d(q_{i_2}, o(P)) \le \sqrt{2} x_o(P)$ since $y_o(P) \ge x_o(P)$.

We then have:
\[
(A + d(p_{i_1}, f) w_{i_1}) d(q_{i_2}, o) w_{i_2} \le 2B x_o(P) w_{i_2} = B (d(q_{i_1}, f) - d(p_{i_1}, f)) w_{i_2}.
\]
Therefore:
\[
r(P, o(P), c) = \frac{A + d(p_{i_1}, f) w_{i_1}}{B}
\le \frac{A + d(q_{i_1}, f) w_{i_1}}{B + d(q_{i_2}, o) w_{i_2}}
= r(Q, o(Q), c).
\]

Equality holds only when $o(P) = o(Q)$. Otherwise, $Q$ has strictly higher ratio.
\end{proof}

\begin{lemma}
For any confidence $c \in [0,1)$ and any point set $P \in P^C_{ca}(c)$ with all $p \in A_x \cup A_{+y}$, there exists $Q \in P^C_{coa}(c)$ with $r(Q, o(Q), c) \ge r(P, o(P), c)$ or a point set $Q'$ such that $r(Q', o(Q'), c) > r(P, o(P), c)$.
\end{lemma}

\begin{proof}
Assume all $p \in P$ are in $\{(-x, 0), (x', 0), (0, x'')\}$. Without loss of generality, scale so that $x'' = 1$. Let $L$, $R$, and $U$ denote the subsets at $(-x, 0)$, $(x', 0)$, and $(0, 1)$, respectively.

Since $f(P, o(P), c) = (0, 0)$, we have $|L| = cn$ and $|R| = |U| = \frac{1 - c}{2}n$.

Moving all points in $U$ to $o(P)$ to form $Q$, by Corollary 7 in \cite{good2020coordinate}, we have $o(Q) = o(P)$.

Let $Q'$ be the weighted COA instance that reaches
\[
r(Q', o(Q'), c) \ge \frac{\sqrt{(1+c)^2 W^2_{\min} + (1-c)^2 W^2_{\max}}}{(1+c) W_{\min}}.
\]

If $r(P, o(P), c)$ is less than this value, we are done.

By Theorem 3, we assume all points in $U$ have weight $W_{\max}$ and those in $L$ and $R$ have $W_{\min}$.

Let $\Delta_f$ and $\Delta_o$ be the reduction in mechanism and optimal costs due to the movement, respectively. Then
\[
\frac{\sqrt{(1+c)^2 W^2_{\min} + (1-c)^2 W^2_{\max}}}{(1+c) W_{\min}} - \frac{\Delta_f}{\Delta_o} \ge 0.
\]

Let $1-y_o$ be the distance moved in the $y$-direction and $\lambda(1-y_o)$ be the distance moved in the $x$-direction. Then we have:

\[
\Delta_f = \frac{(1 - 3c)n}{2} \lambda(1 - y_o) W_{\min} + \frac{(1 - c)n}{2} (1 - y_o) W_{\max},
\]

\[
\Delta_o = \frac{(1 - c)n}{2} \sqrt{1 + \lambda^2} (1 - y_o) W_{\max}
\].

Define function:
\[
g(c, \lambda) = \frac{\sqrt{(1+c)^2 W^2_{\min} + (1-c)^2 W^2_{\max}}}{(1+c) W_{\min}} - \frac{(1 - 3c)\lambda W_{\min} + (1 - c) W_{\max}}{(1 - c) \sqrt{1 + \lambda^2} W_{\max}}
\]

We compute first and second derivatives and find that when $0 < c < 1/3$, $g$ reaches its minimum at
\[
\lambda = \frac{(1 - 3c) W_{\min}}{(1 - c) W_{\max}},
\]

and when $1/3 \leq c < 1$, $g$ achieves its minimum at
\[
\lambda = 0.
\]

Substituting respectively, we find $g(c) > 0$ always, hence
\[
r(Q, o(Q), c) \ge r(P, o(P), c).
\]

Finally, we can shift all the points in $Q$ to the two axes and construct a COA instance.
\end{proof}

\subsection{Robustness Analysis}

We now analyze the robustness of CMP in the weighted setting. Robustness captures how poorly the mechanism can perform when the prediction is arbitrarily wrong. Following the same weighted COA construction as Definition \ref{def:COAW}, we now assume that the prediction \(\hat{o}\) is at \((0,0)\), simulating the wrong prediction. The CMP mechanism still outputs \((0,0)\).

The agents are arranged as before: some are placed at \((0,1)\), and the rest are split symmetrically on \((x,0)\) and \((-x,0)\). Agents with high weight are assigned to the optimal location \((0,1)\), and those with low weight are set at $x$-axis. This setting ensures that the optimal solution incurs low cost, while the output of mechanism takes higher cost due to its reliance on the wrong prediction.

\begin{theorem}
For CMP with confidence \(c \in [0, 1)\), the worst-case robustness in weighted COA instances is
\[ \frac{\sqrt{(1-c)^2 W_{\min}^2 + (1+c)^2 W_{\max}^2}}{(1-c) W_{\min}}. \]
\end{theorem}

\begin{proof}
    To make approximation ratio as big as possible, we should maximize the amount of agents in $(0, 1)$. To guarantee the CMP output $(0, 0)$, there should be $\frac{(1+c)n}{2}$ agents report $(0, 1)$ at most. In this instance, the approximation ratio is 
    \begin{equation*}
        \frac{\frac{1-c}{2}nxW_{\min}+\frac{1+c}{2}nW_{\max}}{\frac{1-c}{2}nW_{\min}\sqrt{1+x^2}} = \frac{\left(1-c\right)xW_{\min}+\left(1+c\right)W_{\max}}{\left(1-c\right)W_{\min}\sqrt{1+x^2}}.
    \end{equation*}

    Take the first derivative with respect to $x$, we have
    \begin{equation*}
        \frac{\left(1-c\right)W_{\min}-\left(1+c\right)xW_{\max}}{\left(1-c\right)\left(1+x^2\right)^{3/2}W_{\min}}.
    \end{equation*}

    Let the numerator be $0$, we get $x = \frac{\left(1-c\right)W_{\min}}{\left(1+c\right)W_{\max}}$. For any $x < \frac{\left(1-c\right)W_{\min}}{\left(1+c\right)W_{\max}}$, the numerator is positive, and for any $x > \frac{\left(1-c\right)W_{\min}}{\left(1+c\right)W_{\max}}$, the numerator is negative, therefore when $x = \frac{\left(1-c\right)W_{\min}}{\left(1+c\right)W_{\max}}$, the approximation ratio reaches its maximum value.

    Therefore, the upper bound of robustness is $\frac{\sqrt{(1-c)^2W^2_{\min}+(1+c)^2W^2_{\max}}}{(1-c)W_{\min}}$.
\end{proof}

This result looks symmetric to the consistency bound, but with the roles of \(c\) and \(-c\) reversed, illustrating the trade-off: as the mechanism places more trust in the prediction (larger \(c\)), consistency improves while robustness deteriorates.

\subsection{Trade-off Between Consistency and Robustness}

Our previous results show that consistency and robustness cannot be optimized at the same time. As the confidence parameter increases, the mechanism relies more on the predicted location, resulting in improved performance when the prediction is accurate. However, this same reliance can lead to degraded performance when the prediction is far from the true optimum.

To make this trade-off concrete, consider a simple instance in $\mathbb{R}^2$ with three agents. One agent with weight $4$ is located at $(0,1)$, while two agents with weight $1$ are placed at $(-1,0)$ and $(1,0)$. The optimal location that minimizes the utilitarian social cost is clearly $(0,1)$, where the highest-weight agent locates.

If the prediction is accurate, $\hat{o} = (0,1)$, then the phantom points introduced by CMP is the optimal point. The larger $c$ is, the more phantom points are introduced, and the mechanism outputs a facility location closer to $(0,1)$. Or more specifically, when $0 < c < 1/3$, no phantom points are introduced, the CMP outputs $(0,0)$, and the consistency is $\frac{3}{\sqrt{2}}$; when $1/3 \leq c < 2/3$, only one phantom point is introduced, the CMP also outputs $(0,0)$, and the consistency is $\frac{3}{\sqrt{2}}$; when $2/3 \leq c < 1$, two phantom points are introduced, the CMP outputs $(0,1)$, and the consistency is $1$. In this case, the consistency decreases when $c$ get larger.

In contrast, if the prediction is wrong, we consider two cases: $\hat{o} = (0,-10)$ and $\hat{o} = (0,1/2)$. In the former case, the CMP always outputs $(0,0)$, which means CMP has a bottommost performance guarantee when the prediction is extremely wrong. In the latter case, the CMP outputs $(0,0)$ when $0 < c < 2/3$, and $(0,1/2)$ when $2/3 \leq c < 1$, in which case the robustness decreases to $\frac{2+\sqrt{5}}{2\sqrt{2}}$. It shows that when the prediction is slightly wrong, CMP can still provide a meaningful robustness guarantee as long as the confidence is high enough.

This example highlights the trade-off: placing more confidence on the predictions improves consistency under accurate predictions, but may worsen robustness when predictions are wrong. No fixed value of $c$ can optimize both objectives simultaneously, showing a limitation of deterministic mechanisms in the weighted setting.

In the next section, we formalize this limitation and prove that no deterministic mechanism can simultaneously achieve perfect consistency and bounded robustness in the weighted setting.

\subsection{Impossibility of Simultaneous Optimal Trade-off}

We now demonstrate a fundamental impossibility result: \emph{No mechanism can simultaneously achieve perfect consistency and bounded robustness} in the weighted facility location problem with prediction. 

Generalized Coordinate-wise Median (GCM) mechanism proposed by Procaccia et al. \cite{procaccia2013approximate} assumes that the locations are in the (low-dimensional) Euclidean space and computes the median of the reported locations along each coordinate. 

With Corollary 3 in \cite{peters1993range}, we know that any deterministic, strategyproof, anonymous, and unanimous mechanism can be expressed as a GCM mechanism with $n-1$ constant points in $P$. 

We construct two specific scenarios in $\mathbb{R}^2$, where all agents lie on the vertical line $x = 0$:

 \emph{Instance A:}
         4 agents at $(0,10)$, each with weight $1$;
         1 agent at $(0,20)$, with weight $5$.
    
\emph{Instance B:}
        4 agents at $(0,10)$, each with weight $5$;
        1 agent at $(0,20)$, with weight $1$.

In Instance A, to reach $1$-consistency, all the $4$ constant point must be at $(x, y)$ where $y > 20$, otherwise, the mechanism would choose $(0, 10)$ and consistency would be $\frac{5}{4}$.

However, in Instance B, the mechanism chooses $(0, 20)$, the robustness become:
\[
\frac{200}{10} = 20.
\]

We can easily imagine that this ratio grows \emph{linearly} with the $\frac{W_{\max}}{W_{\min}}$. More generally, if the agent at $(0,20)$ has weight $W_{\min}$ and the others have weight $W_{\max}$, the robustness becomes:
\[
\frac{(n-1) \cdot W_{\max} \cdot d}{W_{\min} \cdot d} = (n-1) \cdot \frac{W_{\max}}{W_{\min}}.
\]

This example is exactly what Theorem \ref{thm2} says: \emph{There is no deterministic and strategyproof mechanism in the weighted FLP with predictions that is $1$-consistent and  $O\left( n \cdot \frac{W_{\max}}{W_{\min}} \right)$-robust.} It demonstrates that requiring consistency $= 1$ inevitably leads to unbounded robustness as $\frac{W_{\max}}{W_{\min}} \to \infty$. Therefore, no deterministic mechanism can simultaneously ensure perfect consistency and bounded robustness in the presence of adversarial prediction error.

\paragraph*{Acknowledgement.}The work of Yangguang Shi was partially supported at Shandong University by the Science Fund Program of Shandong Province for Distinguished Oversea Young Scholars (Grant No. 2023HWYQ-006), by the National Natural Science Foundation of China (Grant No. 62302273), and by the Shandong Province Sponsored Overseas Study Program (Grant No. 10000082163106).

\bibliographystyle{plain}
\bibliography{mybibliography}

\end{document}